%% file: main.tex
\documentclass{article}
\newif\ifdraft
\drafttrue

\usepackage{amsmath, amssymb,amsthm}
\usepackage{mathtools} 
\usepackage[utf8]{inputenc}
\usepackage[T1]{fontenc}
\usepackage{dsfont}
\usepackage{placeins}
\usepackage{xcolor}
 \usepackage{aliascnt}
\usepackage{authblk}

\usepackage{fullpage}
\usepackage[verbose,colorlinks]{hyperref}

\definecolor{darkred}{rgb}{0.5,0,0}
\definecolor{lightblue}{rgb}{0,0.4,0.8}
\definecolor{darkgreen}{rgb}{0,0.5,0}

\DeclarePairedDelimiter\abs{\lvert}{\rvert}%

\newcommand{\eps}{\epsilon}

\hypersetup{colorlinks=true, linkcolor=blue, urlcolor=blue, citecolor=darkgreen}

\newcommand{\handout}[6]{
   \renewcommand{\thepage}{\Large\bf #2 / Page \arabic{page}}
   \noindent
   \begin{center}
   \framebox{
      \vbox{
    \hbox to 5.78in { {\bf #1}
     	 \hfill #3 }
       \vspace{4mm}
       \hbox to 5.78in { {\Large \hfill #6  \hfill} }
       \vspace{2mm}
       \hbox to 5.78in { {\it #4 \hfill #5} }
      }
   }
   \end{center}
   \vspace*{4mm}
}

\def\NewTheorem#1#2{%
  \newaliascnt{#1}{theorem}
  \newtheorem{#1}[#1]{#2}
  \aliascntresetthe{#1}
  \expandafter\def\csname #1autorefname\endcsname{#2}
}

 \newtheorem{theorem}{Theorem}[section]
\NewTheorem{lemma}{Lemma}
\NewTheorem{corollary}{Corollary}
\NewTheorem{definition}{Definition}
\NewTheorem{remark}{Remark}

\NewTheorem{conjecture}{Conjecture}
\NewTheorem{observation}{Observation}
\NewTheorem{assumption}{Assumption} 
\NewTheorem{assumptions}{Assumptions}

\NewTheorem{proposition}{Proposition}

\NewTheorem{claim}{Claim}

\theoremstyle{remark}

\newcommand{\ripref}[2]{ \it \hyperref[#2]{\it #1}~\ref{#2}    }
\newcommand{\dist}{\operatorname{dist}}

\newenvironment{proof-sketch}{\noindent{\bf Sketch of Proof}\hspace*{1em}}{\qed\bigskip}
\newenvironment{proof-idea}{\noindent{\bf Proof Idea}\hspace*{1em}}{\qed\bigskip}
\newenvironment{proof-of-lemma}[1]{\noindent{\bf Proof of Lemma #1}\hspace*{1em}}{\qed\bigskip}
\newenvironment{proof-attempt}{\noindent{\bf Proof Attempt}\hspace*{1em}}{\qed\bigskip}

\makeatletter
\def\fnum@figure{{\bf Figure \thefigure}}
\def\fnum@table{{\bf Table \thetable}}
\long\def\@mycaption#1[#2]#3{\addcontentsline{\csname
  ext@#1\endcsname}{#1}{\protect\numberline{\csname
  the#1\endcsname}{\ignorespaces #2}}\par
  \begingroup
    \@parboxrestore
    \small
    \@makecaption{\csname fnum@#1\endcsname}{\ignorespaces #3}\par
  \endgroup}
\def\mycaption{\refstepcounter\@captype \@dblarg{\@mycaption\@captype}}
\makeatother

\newcommand{\mathify}[1]{\ifmmode{#1}\else\mbox{$#1$}\fi}
\newcommand{\bigO}O
\newcommand{\set}[1]{\mathify{\left\{ #1 \right\}}}

\newcommand{\remove}[1]{}
\newcommand{\ignore}[1]{}

\newcommand{\floor}[1]{\left\lfloor #1 \right\rfloor}

\newcommand\E[1]{\mathbb{E}\left[\,#1\,\right]}
\renewcommand{\Pr}[1]{\mathbb{P}\left[\,#1\,\right]}

\renewcommand{\eps}{\varepsilon}

\title{How to Color a French Flag\\
{\normalsize Biologically Inspired Algorithms for Scale-Invariant Patterning }}
\author{Alberto Ancona}
\author{Ayesha Bajwa}
\author{Nancy Lynch}
\author{Frederik Mallmann-Trenn}
\affil{Massachusetts Institute of Technology}
\date{}

\begin{document}

\maketitle

\input{abstract}
\setcounter{tocdepth}{1}
\tableofcontents

\vfill

\input{intro}

\input{models}
\input{diffusion}
\input{algorithms}

\input{lowerbounds}

\input{conclusion}
\section*{Acknowledgements}
We thank Ama Koranteng, Adam Sealfon, and Vipul Vachharajani for valuable discussions and contributions.

{ \bibliography{./references}

}

\appendix

\input{auxilary.tex}
\input{message_approx.tex}
\input{missing.tex}

\end{document}

%% file: abstract.tex
\begin{abstract}
    In the \emph{French flag problem}, initially uncolored cells on a grid must differentiate to become blue, white or red. The goal is for the cells to color the grid as a French flag, i.e., a three-colored triband, in a distributed manner. To solve a generalized version of the problem in a distributed computational setting, we consider two models: a biologically-inspired version that relies on morphogens (diffusing proteins acting as chemical signals) and a more abstract version based on reliable message passing between cellular agents.
    
    \medskip 
    Much of developmental biology research has focused on concentration-based approaches using morphogens, since morphogen gradients are thought to be an underlying mechanism in tissue patterning. We show that both our model types easily achieve a {\em French ribbon} - a French flag in the 1D case. However, extending the ribbon to the 2D flag in the concentration model is somewhat difficult unless each agent has additional positional information. Assuming that cells are are identical, it is impossible to achieve a French flag or even a close approximation. In contrast, using a message-based approach in the 2D case only requires  assuming that agents can be represented as constant size state machines. 
    
    \medskip 
    We hope that our insights may lay some groundwork for what kind of message passing abstractions or guarantees, if any, may be useful in analogy to cells communicating at long and short distances to solve patterning problems. In addition, we hope that our models and findings may be of interest in the design of nano-robots.
\end{abstract}


%% file: intro.tex
\section{Introduction}

In the \emph{French flag problem}, initially uncolored cells on a grid must differentiate to become blue, white or red, ultimately coloring the grid as a three-colored triband without any centralized decision-making. Lewis Wolpert's original French flag problem formulation  \cite{Wolpert1968} and   \cite{Wolpert1969} has been applied and extended to understand how organisms determine cell fate, or final differentiated cell type, a question central to developmental biology.
Wolpert's formulation of positional information models is both complementary to and contrasted with Turing's earlier formulation of reaction-diffusion instability \cite{Turing1952}, which relies on random asymmetries that arise from activator-inhibitor dynamics in a developmental system. Our methods make use of both positional information and initial asymmetry. However, we distinguish between absolute and relative positional information to probe whether full knowledge of the coordinates is needed to solve the problem, or if strictly less information suffices.

Broadly speaking, our work is inspired by the biological mechanisms leading to cell fate decisions in the original French flag problem. These long and short-distance mechanisms inform the design of algorithms and analyses of the problem in two distributed computing contexts. More precisely, we relate a reliable message passing model (\autoref{sec:message_model}) with local cell-cell communication and a concentration-based model (\autoref{sec:concentration_model}) with morphogen gradients over long distances.

We analyze a generalized French flag problem for $k$ colors in these two computational models.
We aim to understand the resources and minimum set of assumptions required to solve the problem exactly or approximately. In particular, we study the question of whether cells need to know their exact positions and the grid dimensions in order to solve the $k$-flag problem. We hope that describing and quantifying the resources and information required might have some translation back to biology, and in particular to the mechanisms enabling scale-invariant patterning.


We begin by studying the {\em French ribbon problem}, the 1D scenario in our computational models. Both exact and approximate solutions are possible, with a general tradeoff between precision and space complexity. While both models easily achieve a French ribbon, extending 1D decision-making to the 2D setting is provably difficult in the concentration model. We show that in a 2D grid with point sources at the corners, each agent knowing its absolute distance to every source is insufficient positional information to color the grid even approximately correctly. 
On the other hand, extending to the 2D setting is easy in the message passing model. We analyze numerous efficient algorithms to demonstrate tradeoffs between time complexity, message size, memory size and precision of the obtained French flag.

We do not claim more accurate or thorough models than those proposed by the biology community, especially since downstream protein interactions are abstracted in our models. However, we hope this work may illuminate computational abstractions or guarantees that may be useful in analogy to cells communicating at long and short distances to solve patterning problems. 



\subsection{Relevant Biology Background}
Much biological research has focused on Wolpert's concept of cellular positional information, as we do in our analogy. A key principle of our models is that initial asymmetry and local communication eventually leads to long-distance transmission of the relative positional information of cellular agents, allowing for distributed decision-making.
Morphogens, molecules acting as chemical signals, are thought to underlie cell-cell communication and patterning over long distances.
Exactly how these morphogens produce scale-invariant patterns in organisms and tissues of varying size is an interesting biological question \cite{Gregor2005}.
Various other mechanisms exist for more localized cell-cell communication via cell surface receptors and the ligands that bind to them. These include the transmembrane receptor protein Notch and its membrane-bound protein (ligand) Delta, a system previously studied in a distributed computing context \cite{Afek183}. There are also physical channels such as gap junctions in animal cells and plasmodesmata in plant cells, through which local signalling molecules can pass \cite{alberts2017molecular}.

One well-studied morphogen is {\em Bicoid (Bcd)} for anterior-posterior patterning in {\em Drosophila melanogaster}, or fruit flies \cite{NussleinVolhard1980, Driever1988}. The full {\em Drosophila} patterning is determined not by Bcd alone, but also by other morphogens and a downstream gene network whose expression results in a sequence of decisions. Bicoid and Hunchback together regulate the anterior portions of the embryo while Nanos and Caudal regulate the posterior, activating gap gene expression, in turn activating pair-rule gene expression, and leading to the expression of segment-polarity genes and homeotic selector genes that ultimately specify the body plan \cite{Gilbert2000}. While some work focuses on how regulation and expression of this gene network encodes positional information \cite{Gregor2007, Tkaik2008, Dubuis2013, Petkova2019}, our analogy is to the primary morphogen inputs from maternal effect genes, modeled as external inputs to distributed cellular agents. These inputs are the maternal effect factors -- like Bcd -- which initially carry positional information. In the context of embryonic development, they are proteins synthesized in the embryo that result from direct transmission of the mother's genetic material \cite{Gilbert2000}.

Another well-studied example is {\em Sonic hedgehog (Shh)}, a morphogen for neural patterning in vertebrates, including humans \cite{Patten2000, Dessaud2008}. Graded Shh signalling from an elongated source (the notochord and the floor plate cells) induces concentration-dependent gene expression in the vertebrate neural tube. When extending the patterning problem to a 2D grid, Shh therefore provides one empirical example of how morphogen sources along a single axis or surface can reduce the 2D problem to a 1D problem in principle. Shh is consistent with Wolpert's simple positional information model to some extent, since there is a localized source and concentration determines gene expression. However, that model does not include how cells in the target region consume Shh and thus alter the Shh gradient, in turn altering the morphogen response \cite{Dessaud2008}. Our models leave these more complex interactions and cellular feedback which affect the gradient (as well as signal transduction in downstream genetic regulatory networks) to future work.


\subsection{Related Work on the French Flag Problem}

Building on earlier work on gradients \cite{Morgan1905, Dalaq1938, Stumpf1966}, Lewis Wolpert proposed the French flag problem and model in the late 1960s \cite{Wolpert1968, Wolpert1969}, focusing on the notion of positional information and its generalization to other patterning mechanisms. By receiving information that indicates their relative position, cells in a multicellular organism may form a size-invariant pattern, such as the three stripes of the French flag, by differential gene expression.

Subsequent papers validated the concept of positional information through empirical studies in model species \cite{Summerbell1973, NussleinVolhard1980, Driever1988}.
Crick explored diffusion as a primary mechanism of positional information \cite{Crick1970}. Turing studied reaction-diffusion instability as the driver of morphogenesis \cite{Turing1952}, an earlier paradigm often contrasted with positional information. The simple notion that cells, distributed along a morphogen gradient, may learn positional information via concentration has fundamentally altered the field of developmental biology \cite{Jaeger2009, Green2015}.

Simple early models of positional information in morphogenesis have been critiqued and extended \cite{Wolpert1989, Gordon2016}. The French flag problem has been constructed under various models, including growth and repair simulation models \cite{Miller2004} and reaction-diffusion experimental models \cite{Zadorin2017}. Recent empirical work in \textit{Drosophila melanogaster} explores positional information downstream of morphogen gradients: in gene regulation and expression controlled by diffusing morphogens. It is possible to quantify in information-theoretic terms the amount of information causing downstream decisions, such as decoding cellular identities using positional information from the four gap genes \cite{Dubuis2013, Petkova2019}. Furthermore, empirical measurements on the mechanisms of cellular differentiation validate the information capacity of regulatory elements to be sufficiently complex \cite{Tkaik2008} and show that cells can detect particular morphogens 
with enough precision to determine cellular outputs \cite{Gregor2007}.

To the best of our knowledge, no one has yet explored the original French flag problem in a distributed algorithms setting, with agents analogous to cells. We are able to reduce the complexity of some of our proposed algorithms using previous work on approximate counting algorithms \cite{Morris1978, Flajolet1985}.

\subsection{Results}
We first present our results for the concentration model,
where we assume that each node on a line only has access to morphogens concentrations $c_1$ and $c_2$ emitted from the endpoints of the line.   
We define the model formally in
\autoref{sec:concentration_model}.

On the positive side, it is possible to solve the French ribbon problem exactly.

\begin{theorem}
\label{thm:1dconc}
Algorithm Exact Concentration Ribbon  solves the $k$-ribbon problem in the concentration model for an $n$-process line graph of arbitrary finite length $a$, with constant time and communication complexity, given that processes have knowledge of morphogen concentrations $c_1$ and $c_2$, which have reached steady states, as well as the gradient function.
\end{theorem}

On the negative side, we show that extending to the French flag (2D-case) with just four point-sources at the corners is infeasible.
Here, symmetry prevents us from obtaining a $\varepsilon$-approximate algorithm in this model.

\begin{theorem}
\label{thm:gradientlower}
Consider the concentration model. Fix any $\varepsilon \in (0,1/6)$.
No algorithm can produce an $\varepsilon$-approximate French flag.
\end{theorem}

This is in sharp contrast to the message passing model where even exact solutions are possible.
Our results are summarized in  \autoref{tab:results}.
The exact statements can be found in
\autoref{sec:mess_algos}.

\begin{center}
\begin{table}[ht]
\centering
\begin{tabular}{lllllll}
Algorithm             & Rounds & Memory per Agent       & \# Msgs & Bits per Msg         & Exact         & Reference      \\
\hline
{\small Exact Count}               & $(2-1/k)n$ & $3\log{n} + O(1)$ & $O(n)$ & $O(\log{n})$ & $\checkmark$ & \ripref{Thm.}{thm:simplecounting}\\
{\small Exact Silent Count}        & $3n$ & $2\log{n}+O(1)$         &       $O(n)$       &            $O(1)$          &                 $\checkmark$     &\ripref{Thm.}{thm:orientedribbon}     \\
{\small Bubble Sort}           & $3n$ & $O(\log k)$                    &      $O(n^2)$        &           $O(\log{k})$          &                   $\checkmark$    & \ripref{Thm.}{thm:bubblesort}     \\
{\small Approx Count}        & $2n$ & $2\log\log{n}+O(1)$ & $O(n)$ & $O(\log\log{n})$ & $\times$ & \ripref{Thm.}{thm:approx_count} \\
\end{tabular}
\caption{Comparison of our algorithms in the message passing model. For brevity we ignore additive $O(k)$ terms in the round complexity.}
\label{tab:results}
\end{table}
\end{center}

\FloatBarrier
It turns out that the time complexity of Algorithm {\em Exact Count}
is tight up to an additive $2k$ term, which we show in \autoref{thm:1dorientedlb}, regardless of $k$ and the starting agent.
We would also like to highlight the memory and message complexity of {\em Bubble Sort}, which is independent of $n$ and in fact constant assuming $k=O(1)$.
Finally, we show in \autoref{sec:extend} and \autoref{sec:approx2D} how all of these algorithms can be extended to the 2D case.

%% file: models.tex
\section{Models and Notation}
We now define the models formally. Throughout this paper we assume  the number of colors  $k$ to be constant\footnote{However, for clarity, we sometimes highlight the dependency on $k$.}. 

\subsection{Concentration Model}
\label{sec:concentration_model}
For concentration-based solutions to the French flag problem, we 
assume that each agent
receives concentration inputs from up to four source agents $s_1$, $s_2$, $s_3$, and $s_4$.
The \emph{measured concentration} 
a cell at 2D coordinate
$C=(x,y)$ 
receives from 
source $s_i$, $i\in[4]$ is given by the following \emph{gradient function}
$\lambda_i(C)$ and assume (i) that the function is invertible 
and (ii) that the function is monotonically decreasing in $\dist(C,s_i)$, where
$\dist(C,s_i)$, denotes the distance between cell $C$ and the source $s_i$.
For concreteness, consider the following power-law function
\begin{equation}\label{eq:input}
    \lambda_i(C)= \frac{1}{\dist(C,s_i)^\alpha}
\end{equation}
\vspace{2mm}

 where $\alpha$ is the power-law constant.
 This family of functions is also handy
 for the 1D case with coordinate $C=x$ 
and source $s_i$, $i\in[2]$ in \autoref{sec:concentration}, where we argue that coloring correctly can be reduced to comparing
 $\lambda_1(C)/\lambda_2(C) $
 to $2^\alpha$ and $2^{-\alpha}$.
 
 Though we choose a power-law for convenience, our upper bounds and lower bounds hold for more general gradient functions satisfying constraints $(i)$ and $(ii)$. 
 Deriving precise thresholds for $\lambda_1(C)$ and $\lambda_1(C)$ is more difficult when the thresholds fall close together or when the gradient function is complicated.
The more difficult these conditions, the less biologically practical it may be.
 
 We do not assume any noise, so agents have arbitrarily good precision in measuring concentration. 
 Additionally, we assume that the cells do not receive any other input apart from measured concentration. In particular, they do not have any other positional information such as knowledge of their coordinate or the total ribbon or flag size. We assume that all agents behave identically, performing the same algorithms.



For the French ribbon, we assume that the two sources $s_1$ and $s_2$ are positioned at the ends of the line.
For the French flag we assume the $s_i \in [4]$ to be positioned at the four corners. We make this assumption in order to understand if the concentration model is `strong' enough to solve the French flag problem without any additional communication. Assuming that additional sources are placed at convenient positions such as $(a/3,0)$ for example, defies the idea of of scale invariant systems.
The corner points are already distinguished in that they only have two neighbors, and if one were to place a constant number of sources, these positions are somewhat natural.

\subsection{Message Passing Model}
\label{sec:message_model}
We first consider a 1D version of the French flag problem which we call the \emph{French ribbon problem}.
We assume a line graph consisting of $n$ nodes which we refer to as agents.
We later consider the 2D version, the standard \emph{French flag problem}.
Here, the graph is a $a\times b$ grid on $n=a\cdot b$ agents.
We assume synchronous rounds. In each round, agents can communicate reliably with neighboring agents.

We assume that all agents perform the same algorithm and have no knowledge of their global position, but have a common sense of direction  $dir\in\{up,down,left,right\}$. In particular, agents at the edges know that they are endpoints of rows or columns (or both, if they are corners). Initially, all but one arbitrary agent (the {\em starting agent}) are \emph{asleep}. We refer to this agent as the \emph{source} $s$. 






The goal is to design algorithms that solve the French ribbon problem. Eventually, each agent must output a color so that the line is segmented into three colors: blue, white, and red from left to right. Formally, if $b$, $w$, and $r$ denote the number of agents of each respective color, $\max\{\abs{b-w} , \abs{b-r} , \abs{w-r}\} \leq 1$. In addition, each color should be in a single, contiguous sub-line of $G$---blue, white, red from left to right.
 We also define the more general 1D $k$-Ribbon problem in the same model, in which there are $k$ distinct colors \{1, ..., $k$\} which must form bands of approximately equal size, in increasing numerical order, along a line graph of $n$ agents.

\medskip

A solution to the French flag problem requires that every agent outputs a single color, such that the grid is divided into three vertical blocks. Every row must abide by the requirements of the French ribbon problem, such that the left side is blue and the right side is red. Furthermore, an agent should be the same color as the agent above and below it in its column. The 2D $k$-Flag problem generalizes in the same manner as above. 




\subsection{Notation}
We say a $k$-colored flag  of dimensions $a\times b$ is an \emph{$\varepsilon$-approximate} (French) flag if 
for every color $z \in \{1, ..., k\}$ the following hold.
For each agent $u$ with coordinates $(x,y)$:

\begin{enumerate}
    \item if $x \in \left[ (\frac{z-1}{k}+\varepsilon)\cdot a, (\frac{z}{k}-\varepsilon)\cdot a  \right] $, then the agent has color $z$.
    \item if $u$ has color $z$, then  $x \in \left[(\frac{z-1}{k}-\varepsilon)\cdot a, (\frac{z}{k}+\varepsilon)\cdot a  \right]$.
\end{enumerate} 

Intuitively speaking, the definition ensures two properties. First, agents that are clearly within one stripe should have the corresponding color.
Second,  agents that are close to a color border $(c_1,c_2)$  should have either color $c_1$ or $c_2$.



%% file: diffusion.tex
\section{Concentration  Model Results}\label{sec:concentration}

\subsection{1D Exact Concentration Ribbon}

\subsubsection*{Algorithm {\em Exact Concentration Ribbon}} We consider an $n$-process line of arbitrary finite length $a$ in the concentration model. Assume morphogens $m_1$ and $m_2$ (with concentrations $c_1$ and $c_2$) are each secreted by one of the endpoint processes.
We assume the underlying gradient function for concentration given position $x$ is the inverse power law in $\alpha$, which is assumed to be noiseless.

Assume that $m_1$ is secreted at $x=0$ and $m_2$ is secreted at $x=a$, we have $c_1 = 1/x^\alpha$ and $c_2 = 1/(a-x)^\alpha$. The ratio of $c_2$ to $c_1$ is then $(a-x)^\alpha / x^\alpha$. 
Each agent computes this ratio independently from measured the values of $c1$ and $c2$. Let {\em ratio} $= c_2/c_1$.
After calculating its measured ratio, each agent computes the smallest color $z$ such that {\em ratio} $\geq ((z-1)/(k-z))^\alpha$, decides color $z$, and halts.

Importantly, the algorithm is size-invariant. We prove below that it holds for any line graph of arbitrary finite length.


\begin{proof}[Proof of \autoref{thm:1dconc}]$ $

To prove correctness of Algorithm {\em Exact Concentration Ribbon}, consider a particular agent $p$ on a line graph of fixed length $a$. Note that the $t_z$ are strictly  increasing. Let $t_0=0$.
Suppose the color agent $p$ should get is $z^*$.
For $z^* < k$, it thus suffices to show that 
$c_2/c_1 \geq t_{{z^*}-1}$ and $c_2/c_1 \leq t_{z^*}$ using that the $t_z$ are strictly monotonically increasing.
For $z^*=k$ we only require the first condition and therefore we assume henceforth $z^*<k$.
We have
$c_2 \geq 1/\left(\frac{k-z^*-1}{k}a \right)^\alpha $
and 
$c_1 \leq 1/\left(\frac{z^*}{k}a \right)^\alpha $, giving
$ \frac{c_2}{c_1} \geq \frac{\left(\frac{z^*}{k}a \right)^\alpha}{\left(\frac{k-z^*-1}{k}a \right)^\alpha} 
= \left(\frac{z}{k-z-1}\right)^\alpha =t_{{z^*}-1}.
$

Moreover, 
$c_2 \leq 1/\left(\frac{k-z^*}{k}a \right)^\alpha$.
and
$c_1 \geq 1/\left(\frac{z^*-1}{k}a \right)^\alpha$.
Hence, 
$ \frac{c_2}{c_1} \leq \frac{\left(\frac{z^*-1}{k}a \right)^\alpha}{\left(\frac{k-z^*}{k}a \right)^\alpha}
\leq
 \left(\frac{z-1}{k-z}\right)^\alpha =t_{{z^*}}.
$
This completes the correctness proof.



%
Finding color $z$ takes constant time since we assume $k=O(1)$, and constant space to store the measured concentrations, calculate the concentration ratio, and calculate thresholds.

Note that there is no message passing in this model. Therefore in the synchronous model, the time complexity of this algorithm is $O(1)$ assuming the morphogen concentrations have reached steady state.
\end{proof}


\subsection{2D Concentration Lower Bound}

\begin{figure}[h]
    \centering
    \includegraphics[width=\textwidth]{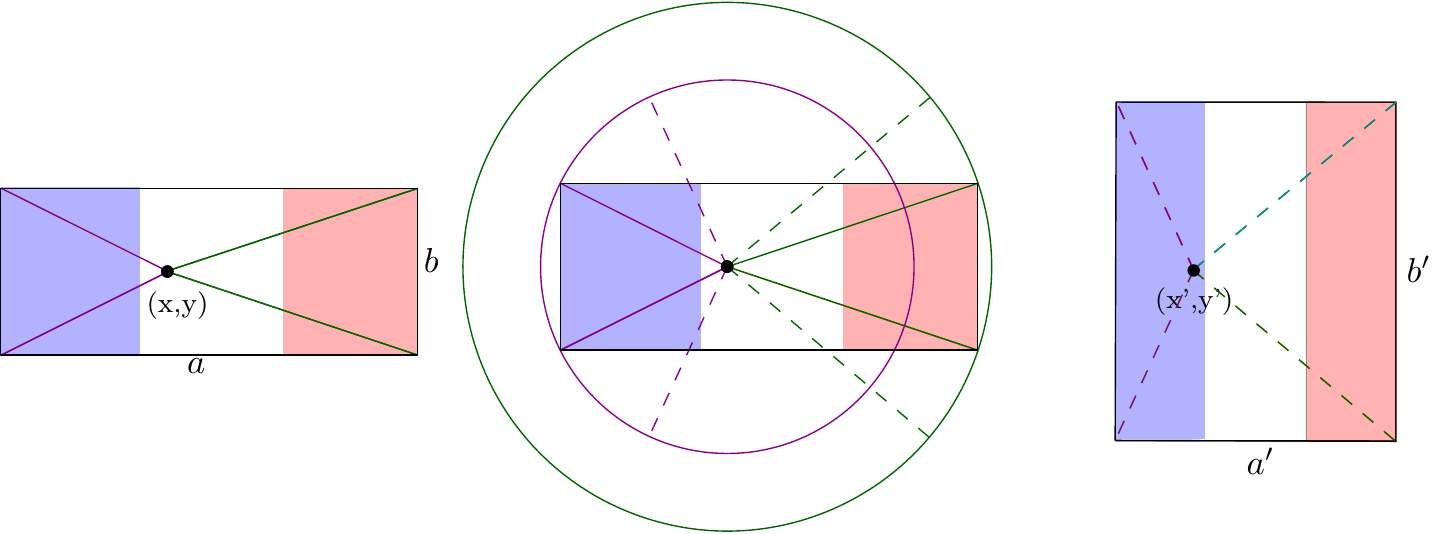}
    \caption{The l.h.s. depicts an arbitrary original flag. In the proof of \autoref{thm:gradientlower} we argue how to construct a new flag (r.h.s.) such that  there are two agents in both flags that 1) have exactly the same distances from the respective sources and 2) must choose different colors.
    Since the two agents have the same respective distance to every source, they receive the same concentration input and cannot distinguish between the settings, making it impossible to always color correctly. We construct the new flag by changing the aspect ratio in a way that maintains the distances. The figure in the middle depicts this transformation.  
    }
    \label{fig:diffusion}
\end{figure}

In this section we prove \autoref{thm:gradientlower}, showing that the concentration model, without absolute positional information, cannot produce a correct French flag (or even a good approximation) regardless of the gradient function.
The idea of the proof is as follows.
Given an arbitrary flag $G$ of dimensions $a\times b$,
we show that we can construct a flag $G'$ with dimensions $a'\times b'$ such that
there are two agents in both flags that 1) have exactly the same distances from the respective sources and 2) must choose different colors.
    Since the two agents have the same respective distance to every source, they receive the same concentration input and cannot distinguish between the settings, making it impossible to always color correctly. 
    See \autoref{fig:diffusion} for an illustration. 
    To show that such a flag $G'$ exists, we frame the constraints as a system of equations and we show that there exists a valid solution.
A formal proof can be found in \autoref{sec:missing}.

%% file: algorithms.tex
\section{Message Passing Model: Exact Coloring}\label{sec:mess_algos}

Before we present our algorithms, note there is a trivial algorithm that works as follows for $k=3$.
The starting agent sends a wakeup message to the leftmost and rightmost agents.
Then start a counter from each of these agents. When a agent receives the counters $c_{left}$ and $c_{right}$, it can determine in which stripe it is
by testing whether 
$c_{left}/c_{right}\geq 2$ or $c_{left}/c_{right}\leq 1/2$. This idea generalizes to arbitrary $k$.

 The algorithms we present improve on the trivial algorithm in various ways. \autoref{tab:results} summarizes the tradeoffs of our approaches in the message passing model. As a starting point, we observe that each agent can learn the number of agents to its left and right, from which information it can determine its own color \cite{Wolpert1969}. This principle is central to some of our algorithms.

\begin{observation}
\label{obs:countobvs}
An agent in the $k$-ribbon problem may determine its correct color knowing the number of agents on both sides of it on the line, and knowing which side should be color 1.
\end{observation}




\subsection{Exact Counting}

\subsubsection*{Algorithm {\em Exact Count}}
The starting agent stores the value $n_{mid}\leftarrow0$ and sends $n_{mid}+1$ in both directions.
Intuitively, the value measures the distance to the starting agent.
All other agents upon waking store the received value as $n_{mid}$ and forward the value $n_{mid}+1$ to the next agent in the same direction. Each agent also stores $t\leftarrow n_{mid}$ and increments $t$ every round after.

When the left endpoint receives a value for $n_{mid}$, it decides on color 1 and sends $n_\ell = 1$ to its right neighbor. The right endpoint does the same, but with color $k$ and in the opposite direction. Each agent stores $n_d$ for either direction $d\in\set{\ell,r}$
which is the number of agents to the left (right, respectively).
Upon receiving $n_d$, the agents forwards $n_d+1$ in the same direction.

After an agent receives both $n_\ell$ and $n_r$, it decides its color using \autoref{obs:countobvs}. In order to get an improved time complexity, an agent may also decide early: if an agent has a value $n_d$ and $t\geq2((k-1)\cdot n_d)-n_{mid}$, it should decide color 1 if $d$ is $\ell$ or color $k$ otherwise.

\begin{theorem}
\label{thm:simplecounting}
Algorithm Exact Count solves the  $k$-ribbon problem and requires at most $(2-\frac{1}{k})\cdot n + k$ rounds, $(4-\frac{2}{k})\cdot n\log{n}$ message bits, and $3\log{n} + \log{k} + O(1)$ bits of memory per agent.
\end{theorem}
\begin{proof}[Proof of \autoref{thm:simplecounting}]$ $

We first address correctness. Any received value of $n_\ell$ or $n_r$ is the number of agents to the left or right of the agent, because they start at 1 for agents one unit away from the endpoints, and increment as they are passed along the line. If an agent halts early, then $t\geq2((k-1)\cdot n_d)-n_{mid}$ for that agent. Note that it only takes $2(n_{\bar{d}})-n_{mid}$ time for the $n_{\bar{d}}$ message to reach an agent, so $n_{\bar{d}}>2((k-1)\cdot n_d)$, indicating that the agent is in the $d$-most band of the ribbon and thus decides correctly.

We next show round complexity, in two cases. Consider the case in which the starting agent has at least $n/k$ agents to its left and right. Then both endpoints will wake up within at most $(1-\frac{1}{k})\cdot n$ rounds, and the $n_d$ messages will propagate across the line in $n$ additional rounds, after which all agents will have values for $n_\ell$ and $n_r$. Thus, the total number of rounds will be at most $(2-\frac{1}{k})\cdot n$.

Consider the other case, when the starting agent has $\hat{n}_d \leq n/k$ agents in direction $d$. After $(2-\frac{1}{k})\cdot n$ rounds, agents at distance at most $(2-\frac{1}{k})\cdot n-\hat{n}_d$ from the endpoint in the $d$ direction will have a value for $n_d$, and agents at distance at most $(2-\frac{1}{k})\cdot n-(n-\hat{n}_d) = (1-\frac{1}{k})\cdot n+\hat{n}_d$ from the other direction $\bar{d}$ will have a value for $n_{\bar{d}}$. Note that since $\hat{n}_d\leq n/k$, all agents must then have a value for $n_d$, because $(2-\frac{1}{k})\cdot n-\hat{n}_d\geq (2-\frac{2}{k})\cdot n \geq n$. The only agents that do not also have a value for $n_{\bar{d}}$ are those with at least $(1-\frac{1}{k})\cdot n+\hat{n}_d$ agents in direction $\bar{d}$, or equivalently those with $n_d\leq n-((1-\frac{1}{k})\cdot n+\hat{n}_d) = n/k - \hat{n}_d$. But for these agents, if $t\geq 2((k-1)\cdot n_d)-n_{mid}$, they have already halted. Since $n_{mid} = \hat{n}_d-n_d$, we have $2((k-1)\cdot n_d)-n_{mid} = 2((k-1)\cdot n_d)-\hat{n}_d+n_d = (2k-1)\cdot n_d - \hat{n}_d\leq (2k-1)(n/k-\hat{n}_d)-\hat{n}_d \leq (2k-1)n/k = (2-\frac{1}{k})\cdot n$. Thus, these agents halt early after at most $(2-\frac{1}{k})\cdot n$ rounds.

We finally argue message bit complexity and space requirements. At any round, only two messages may be sent, each with values at most $n$; thus, after $(2-\frac{1}{k})\cdot n$ rounds, at most $2\log{n}((2-\frac{1}{k})\cdot n) = (4-\frac{2}{k})\cdot n\log{n}$ message bits are sent. Each agent needs to store three values in $\Theta(n)$, which takes $3\log{n} + O(1)$ bits. Each agent must also store $k$ using $\log{k}$ bits. 
\end{proof}

\subsection{Exact Silent Counting}

The message passing model is reliable, so we improve the message bit complexity in \autoref{thm:simplecounting} by using the lack of a message as information. We give an algorithm that uses silence at a small cost to round complexity.

\subsubsection*{Algorithm {\em Exact Silent Count}}

The starting agent sends the message 0 to the left and 1 to the right. If it is an endpoint, the starting agent sends a 0 and a 1 in the same, 2-bit message to its neighbor. Agents will forward any received messages in the same direction, except endpoints which will send the messages back.

The agents do additional processing. The endpoint on the $d$ side sets $n_d\leftarrow0$ upon waking and never modifies it. Otherwise, the first time an agent receives a message from direction $d$, it sets $n_{\bar{d}}\leftarrow0$, and each round thereafter the agent increments $n_{\bar{d}}$, until it receives a message from the $\bar{d}$ direction, at which point it stops incrementing $n_{\bar{d}}$ and sets $n_{\bar{d}}\leftarrow n_{\bar{d}}/2$. When an agent has final values for $n_\ell$ and $n_r$, and has sent 0 to the left and 1 to the right, it decides its color based on its stored values of $n_\ell$ and $n_r$ using \autoref{obs:countobvs} and halts.

\begin{theorem}
\label{thm:orientedribbon}
Algorithm Exact Silent Count solves the  $k$-Ribbon problem and requires $3 n$ rounds, $6 n$ message bits, and $2\log{n} + \log{k} + O(1)$ bits of memory per agent.
\end{theorem}

\begin{proof}[Proof of \autoref{thm:orientedribbon}]$ $

We show correctness for the case when the starting agent is not an endpoint; we leave that end-case for the reader. W.l.o.g. consider an agent that first receives a 0 from the right. After $2n_\ell$ rounds, the 0 bit will return to the agent after having been forwarded to the left endpoint and back, so the stored value of $n_\ell$ at the end of the round will be correct. After $2n_r$ more rounds, the 0 bit is received again from the right and $n_r$ is correctly set. Thus, as long as the agent receives the 0 bit 3 times, it will color itself correctly. The 0 bit must then travel from the starting agent to the left, back to the right endpoint, then back to the left endpoint; at that point, all agents to the left of the starting agent will correctly color themselves. As long as the agents to the right of the starting agent return the 0 bit leftward, this will occur. We thus have correctness, because all agents only halt after forwarding the opposite bit back to the other side. The same argument applies to the 1 bit in the other direction.

The bits travel at most 3 times fully across the line each, so all agents will terminate after at most $3n$ rounds. Each round has only 2 bits sent, so the algorithm has message bit complexity $6n$. Each agent stores $k$ and two values in $\Theta(n)$, which requires only $2\log{n}+\log{k}+O(1)$ bits of memory each.
\end{proof}

\subsection{Bubble Sort Approach}
Here we show how to use bubble sort to color the flag. Assume blue, white and red are 1, 2, and 3 respectively.

\subsubsection*{Algorithm {\em Bubble Sort}}
The idea of the algorithm is to color each agent alternating with the colors of the flag, to ensure correct total counts of each color without knowing the ribbon length upfront. The algorithm performs swaps in parallel to ensure that blue elements ripple to the left, white elements to the middle, and red elements to the right.
In order to avoid cases in which a agent would like to swap its color with both neighbors at same time, we also ensure through message passing that each agent knows whether it is at an odd or even position and whether the current round is odd or even.
In an even round, any agent at an even position swaps the value (color) with its right neighbor if the right neighbor has a larger value. Odd rounds are analogous.  

\medskip

For the analysis, it seems rather complicated to define potentials that describe the color distribution of the elements at any round in order to guarantee enough progress in every round. Instead, we use the following trick in the analysis. 
Instead of coloring each agent alternating with blue, white and red, we `color' each agent with a unique ID such that all blue agents have smaller IDs than all white agents and all white agents have smaller IDs than all red agents.
We can then reduce the problem to a distributed bubble sort on $n$ distinct elements, allowing for an elegant proof by induction.

\begin{theorem}
\label{thm:bubblesort}
Algorithm Bubble Sort solves the 1-D $k$-Ribbon problem and requires at most $3n$ rounds, $n^2\log k$ message bits, and $O(\log k)$ bits of memory per agent.
\end{theorem}
\begin{proof}
By the algorithm, we have that the $i$th agent chooses blue if  $i\mod 3 = 0$, white if $i\mod 3=1$, and red otherwise.
%
Consider a `new' process $\mathcal{P}'$ 
 in which 
 the $i$th agent gets the color
\vspace{1mm} 
\begin{equation*}
\begin{cases}
\floor{i/3} & \text{ if $i\mod 3 = 0$  }\\
n/3 + \floor{i/3}  & \text{ if $i \mod 3 = 1$}\\
2n/3 + \floor{i/3}   & \text{ otherwise  }\\
\end{cases}
\end{equation*}
\vspace{1mm}

Observe that when we replace the different IDs by the colors, we obtain the original process $\mathcal{P}$. Therefore, it holds that once $\mathcal{P}'$ terminates, so does $\mathcal{P}$.
The advantage of $\mathcal{P}'$ is that due to different values for agents, we can prove
that if we reset time at round $n$,
%
we can show by an induction  on $i$ and second induction $t$ that

\begin{equation}\label{eq:fox}
p_i^t \leq \max\{ i, n-t+2i \},
\end{equation}
\vspace{1mm}

where $p_i^t$ denotes the position of the agent with ID $i$ at $t$. 
Fix an $i$ and assume the claim holds for the first $i-1$ agents and for all $t$. 
We distinguish between two cases.
If for all $j \leq i-1$ we have $p_i^{t-1} \geq p_j^{t-1}+2$, then agent $i$ will move one position to the left

\begin{equation*}
    p_i^t \leq p_{i}^{t-1} - 1 \leq \max\{ i, n-(t-1)+2i \} - 1 \leq \max\{ i, n-t+2i \}.
\end{equation*}
\vspace{1mm}

Otherwise, note that $ p_{i}^{t-1}$ can be at most $1$ (moving to the right) plus the the maximum  position of all IDs $j\leq i-1$ at round $t-1$, which is
$n-(t-1)+2(i-1)$. We have

\begin{equation*}
    p_i^t \leq \max_{j \leq i-1} p_{j}^{t-1} + 1 \leq 
 \max\{ i-1, n-(t-1)+2(i-1) \} + 1 \leq
 \max\{ i, n-t+2i \}.
\end{equation*}
\vspace{1mm}

In either case the inductive step holds and thus \eqref{eq:fox} holds.
Note that the agent with ID $i$ has to be at its correct position ($i$) 
at time $t$ satisfying $n-t+2i = i$, i.e., when
$t= n+i$.
Therefore the process terminates at time $t=2n$.
Since we shifted time by $n$ (starting only onces all agents are assigned a color), the total time is $3n$. 
The bound on the memory and message size follow trivially.





\end{proof}

\subsection{Extending from  Ribbon to Flag}\label{sec:extend}
We may solve the $k$-flag problem by extending any $k$-ribbon algorithm, with little loss in most parameters.


\subsubsection*{Algorithm {\em Up \& Down}}
The starting agent simply begins the $k$-ribbon algorithm for its row, and all agents in the row follow the algorithm to completion once awakened. However, after deciding on a color but before halting, each agent in the row sends a message to its neighbors above and below with its decided color. When an agent is awoken with a color, it decides on that color and forwards the color either above or below before halting.

\begin{theorem}
\label{thm:2d_extend_ribbon}
Given an algorithm for the $k$-ribbon problem which takes $T(n,k)$ rounds, $M(n,k)$ total message bits, and $S(n,k)$ bits of memory per agent, Algorithm  Up \& Down solves the $k$-flag problem on a $a\times b$ grid with at most $a + T(b,k)$ rounds, $ab\log{k} + M(b,k)$ total message bits, and $S(b,k)$ bits of memory per agent.
\end{theorem}
\begin{proof}[Proof of \autoref{thm:2d_extend_ribbon}] $ $
Clearly, the algorithm takes $T(b,k)$ rounds, $M(b,k)$ message bits, and $S(b,k)$ bits of memory per agent just to complete the $k$-ribbon algorithm on the starting row. Subsequently, each agent sends messages of size $\log{k}$ up and down its column, which takes $a$ additional rounds and $a\log{k}$ additional message bits per column, accounting for the added round and message bit complexity. The size of each agent stays the same, because $S(b,k)\geq\log{k}$ simply to represent the color the agent decides.
\end{proof}

We note there are other reductions to the $k$-ribbon problem that optimize more for round complexity rather than space and message bit complexity. For brevity, we leave these to the reader. 

\subsection{Lower Bounds}

We show straightforward lower bounds for 1D and 2D cases in the message-passing model. 


\begin{theorem}
\label{thm:1dorientedlb}
No algorithm exists that can solve the $k$-Ribbon problem on an oriented line graph if all agents are identical, even if endpoints know that they are endpoints, in less than $(2-\frac{1}{k})\cdot n  - 3$ rounds.
\end{theorem}

\begin{proof}[Proof of \autoref{thm:1dorientedlb}]$ $
Suppose such an algorithm $A$ exists. Consider a line of length $n$, where the starting agent is the leftmost agent. Let $p$ be the leftmost agent $p$ that should choose color 1. It shall decide on color 1 in less than $(2-\frac{1}{k})\cdot n-k$ rounds.

Suppose $k$ additional agents were attached on the right side of the graph before the algorithm began, so that there are $n+k$ agents in the line. In this new line graph, $p$ must decide color 0, rather than color 1; otherwise, either there will be two more color 1 or two more color 2 agents than color 0 agents, and $A$ would not solve the French ribbon problem.

In order to distinguish between these two cases, $p$ must receive information from a node to the left of the $n$th agent. The $n$th agent will wake up at round $n-1$, and the information that it is not an endpoint must then propagate to $p$, which is at distance at least $(1-1/k)n-k+1$ from the $n$th agent. Therefore the information does not reach $p$ until round $(2-\frac{1}{k})n-k$. However, $p$ halted at an earlier round by our assumption, and therefore did not distinguish the two cases.
\end{proof}

\begin{theorem}
\label{thm:2dorientedlb}
No algorithm exists to solve the $k$-flag problem on an $a\times b$ grid in less than \\ $\max\set{(2-\frac{1}{k})\cdot b-k, a+b-2}$ rounds.
\end{theorem}

\begin{proof}[Proof of \autoref{thm:2dorientedlb}]$ $
Assume for contradiction that we could solve the problem in less than $(2-\frac{1}{k})\cdot b-3$ rounds. Then we could solve the $k$-ribbon problem in less than $(2-\frac{1}{k})\cdot n-3$ rounds by setting $a = 1$ and $b=n$ and solving the $k$-flag problem.

The algorithm must also take at least $a+b-2$ rounds. Consider the case where the starting agent is in the top-left corner of the grid. Then the bottom-right agent will not wake and decide its color before it receives a message from the starting agent, which is at distance $a+b-2$.
\end{proof}

%% file: lowerbounds.tex

%% file: conclusion.tex
\section*{Conclusion}
In this paper, we demonstrate that the 1D French ribbon problem can be solved exactly and approximately in both the concentration model and the message passing model. However, the 2D French flag problem requires additional positional information in order to satisfy size invariance.

One direct extension of this work is a randomized version of the {\em Silent Count} algorithm (\autoref{thm:orientedribbon}).
An exciting new research direction is how other pattering problems can be solved in more general settings  and under the influence of noise.
Future work could develop models that better capture important biological constraints or shed light on other problems in developmental biology.
For example, one could study models in which part of an organism (e.g., a finger or the beak of a bird) grows over time.


%% file: auxilary.tex


%% file: message_approx.tex
\section{Message Passing Model: Approximate Coloring}

Here, we use the approximation approach of Morris \cite{Morris1978} and Flajolet \cite{Flajolet1985} to reduce space complexity in exchange for a slight increase in error for the final $k$-ribbon. The randomized modification is made to our deterministic exact counting algorithm.

The following theorem gives the guarantees of each counter.
\begin{theorem}[\cite{Flajolet1985}]\label{thm:flajolet}
Consider the counter procedure of \cite{Flajolet1985}.
Using $\log\log{n} + \delta$ bits and letting $\beta = 2^{2^{-\delta}}$, the expectation of the $\log\log{n}+\delta$-bit counter is $\log_\beta((\beta-1)\cdot n + \beta)$, and the value of $n$ we could recover from the counter has standard deviation at most $n/2^{-\delta}$.
\end{theorem}

\subsubsection*{Algorithm {\em Approximate Count}}

The starting agent sends a bit in either direction to wake all agents. When the endpoint in the $d$ direction wakes up, it sets a counter $c_d$ to 0, increments it as in \cite{Flajolet1985}, and sends the resulting value to its neighbor. Each agent upon receiving a message from direction $d$, stores the message as $c_d$, increments it in the same way and forwards the result to the next agent. 

When an agent has received two values of $c_d$, it does the following: For each $i$ in the sequence $1,\dots,k$, if $c_\ell-c_r \leq \log_\beta\frac{i}{k-i}$, then the agent decides on color $i$. If the agent has not decided on a color yet after all $i$, the agent decides on color $k$. After deciding on a color, the agent halts.

\begin{theorem}
\label{thm:approx_count}
Fix any $k$. For $n$ large enough, Algorithm Approximate Count solves the $\eps$-approximate $k$-Ribbon problem for constant $\eps<\frac{1}{2(k-1)}$ with probability $1-\frac{1}{32k}$ and requires $2n$ rounds, $O(n\log\log{n})$ total message bits, and $2\log\log{n} + O(1)$ bits of memory per agent.
\end{theorem}

We restrict $\eps<\frac{1}{2(k-1)}$ because otherwise the color thresholds would bleed into each other and we would have regions with more than two valid colors. The core idea of using an approximate counter as proposed in \cite{Flajolet1985} is that when subtracting the counter from the left and from the right, we get for some $\beta$, ignoring small standard deviations,

\[ \log_\beta( (\beta-1)n_\ell + \beta  )
-\log_\beta( (\beta-1)n_r + \beta  ) \approx\log(n_\ell/n_r).
\]
\vspace{1mm}

Using thresholds for each color then gives the right color. Using monotonicity of the counters, we only need to consider $O(k)$ different counters which allows us to take a union bound over $O(k)$ of them, showing that all $n$ counters are `correct'.

\begin{proof}[Proof of \autoref{thm:approx_count}]$ $
Clearly the algorithm terminates in at most $2n$ rounds. Our analysis below only works if agents have at least $2\log\log{n}+ O(1)$ bits of memory but does not require more than that. Only the approximate counts, which must fit in memory, are sent as messages, and only 2 messages are sent per round, which yields our message bit complexity.

\begin{lemma}
\label{lem:thresh}
The algorithm yields an $\eps$-approximation if the $2(k-1)$ "threshold" agents within distance $\eps n$ of a color threshold decide on the correct color.
\end{lemma}
\begin{proof}
Note that the count from the left monotonically increases from left to right, and similarly for the count from the right from right to left. If the two threshold agents in the $i$th color band color themselves correctly, by the monotonicity of the counts, all colors between them will be $i$ as well, and all agents close to a threshold will only be one of the two bordering colors.
\end{proof}

By \autoref{thm:flajolet}, using at least $\log\log{n} + \delta$ bits for each counter and letting $\delta = \log(1/\eps)+2\log(k)+8$ makes the standard deviation of the value recoverable from the counter be $\sigma \leq \eps n_d/256k^2$ after $n_d$ increments, where $n_d$ is the number of other agents in the $d$ direction. Using Chebyshev's inequality and letting 
\[\beta = 2^{2^{-\delta}},\] the probability that the counter will store a value within $\log_{\beta}{((\beta - 1)(1 \pm \eps/16k)n_d + \beta)}$ after $n_d$ increments is at least $1-\frac{1}{128k^2}$.

Consider the agent at distance exactly $\eps n$ to the left of color border $i$. We require for this agent that $\frac{i-1}{k-i+1} \leq n_\ell/n_r\leq \frac{i}{k-i}$, where $n_\ell = (\frac{i}{k}-\eps)n$ and $n_r = (\frac{k-i}{k}+\eps)n$. We will show that this holds even with our approximate counters. We assume that the following inequality holds, and prove it later.

\begin{equation}\label{eq:nasty}
 \frac{n_\ell}{n_r} (1-5\varepsilon') \leq 
\frac{(\beta-1)n_\ell(1\pm\varepsilon')+\beta}{(\beta-1)n_r(1\pm \varepsilon') +\beta} \leq \frac{n_\ell}{n_r} (1+5\varepsilon'), 
\end{equation}
\vspace{1mm}

We set $\varepsilon' = \frac{\eps}{16k}$. This implies that $\log_{\beta}(\frac{n_\ell}{n_r} (1-5\eps/16k)) \leq c_\ell-c_r \leq \log_{\beta}(\frac{n_\ell}{n_r} (1+5\eps/16k))$. For our agent above, we have $n_\ell / n_r = \frac{i-k\eps}{k-i + k\eps}$, implying $\log_{\beta}(\frac{i-k\eps}{k-i + k\eps} (1-5\eps/16k)) \leq c_\ell-c_r \leq \log_{\beta}(\frac{i-k\eps}{k-i + k\eps} (1+5\eps/16k))$. Thus we need only verify that $\frac{i}{k-i} \geq \frac{i-k\eps}{k-i + k\eps} (1+5\eps/16k)$ and $\frac{i-1}{k-i+1} \leq \frac{i-k\eps}{k-i + k\eps} (1-5\eps/16k)$, which, noting that $\eps < \frac{1}{2(k-1)}$, can easily be shown. A similar argument shows the same result for agents up to distance $\eps n$ to the right of a border line.

With our choice of $\delta$, we thus succeed at each threshold agent with probability at least $1-\frac{1}{64k^2}$ (probability both counters succeed). The probability that all succeed is, by union bound, at least $1-\frac{2(k-1)}{64k^2} \geq 1 - \frac{1}{32k}$.





It only remains to prove \eqref{eq:nasty}, which we do in the following.
To show this we assume that the following inequality holds, which we will  prove later.
\begin{equation}\label{eq:beta}
    \beta \leq \varepsilon' (\beta-1) \min\{n_\ell, n_r\}.
\end{equation}
\vspace{1mm}
This implies, 
\[ \frac{(\beta-1)n_\ell(1\pm\varepsilon')+\beta}{(\beta-1)n_r(1\pm \varepsilon') +\beta} \leq \frac{(\beta-1)n_\ell(1+ 2\varepsilon')}{(\beta-1)n_r(1- \varepsilon')} \leq \frac{n_\ell}{n_r} (1+5\varepsilon'),\]
where we used that $\varepsilon' \in (0,2/5]$.
Similarly, using that $\varepsilon' \geq 0$, we get
\[ \frac{(\beta-1)n_\ell(1\pm\varepsilon')+\beta}{(\beta-1)n_r(1\pm \varepsilon') +\beta}
\geq
\frac{(\beta-1)n_\ell(1-\varepsilon')}{(\beta-1)n_r(1+2 \varepsilon')} \geq  \frac{n_\ell}{n_r} (1-5\varepsilon').
\]
Therefore, \eqref{eq:nasty} holds and it remains to prove \eqref{eq:beta}.
%
%
%
%
%
%
We note that \[\beta-1 =  \Omega(1),\] and so we have
\begin{align*}\eps'(\beta-1)\min(n_r,n_\ell) &= \frac{\eps}{16k} (\beta-1) \min(n_r,n_\ell)\\
 &\geq \frac{\eps}{16k} (\beta-1) \frac{n}{2k}\\
 &> \beta.  \end{align*}
Therefore \eqref{eq:beta} holds, which completes the proof. 

\end{proof}







\subsection{Extending from Ribbon to Flag for Randomized Algorithms}\label{sec:approx2D}
In this section we show how to extend {\em Approx Count} ensuring that round complexity is $O(n)$ and memory and message bit complexity is $O(n\log \log n)$. Of course, one could run the 1D algorithm independently in every row. However, this will result in many errors, which can be avoided in an elegant way.

\subsubsection*{Algorithm {\em Boost}}
Instead of simply copying the color of an agent through its whole column, we can boost the probability of assigning the correct color using the following algorithm.
For any column (in parallel) count each color separately by using message passing, starting at the top-most node of the column. Once a color reaches $T= 72 \log n$, declare it the winner, stop counting, and inform all other nodes on the column.

\begin{theorem}
\label{thm:2d_extend_approx}
Algorithm Boost solves the $\eps$-approximate $k$-flag problem for any constant $\eps< 1/2(k-1)$ with probability $1-1/n$ and requires $3n$ rounds, $O(n\log\log{n})$ total message bits, and $O(\log\log{n})$ bits of memory per agent.
\end{theorem}

\begin{proof}[Proof of \autoref{thm:2d_extend_approx}] 
First observe that at most $k \cdot (\log\log n + 72)+O(1)=O(\log\log n)$ bits memory are necessary to implement the algorithm. 

Fix an arbitrary column. 
Observe that an agent that is `$\eps$-far' from a border is colored correctly w.p. at least $1-1/(6k)$ by \autoref{thm:approx_count}. 



Consider the time step $t=k \cdot T$ after which $k \cdot T < n$ rows were considered. Fix an incorrect color $c'$.
Let $X_i=1$ if the row $i$ has that incorrect color, otherwise $X_i=0$.
Let $X= \sum_{i\leq t} X_i$. Note that each row is correct independently.
We have $\E{X}=t/(6k)=T/6= 12 \log n$.
Thus, using Chernoff bounds, it holds that with probability at least $1-n^4$ the color $c'$ was not declared winning:

\[\Pr{ X \geq T} \leq \Pr{X \geq 2\E{X}}\leq \exp\left(- \frac{\E{X}}{3} \right) = \frac{1}{n^4}.  \]

Taking union bound over all incorrect colors shows that after $t$ steps no incorrect color can have been declared winner. 
First assume that the agent is at least $\eps$-far from all color borders.
Observe that there are at  $k-1$ incorrect colors and hence, by a pigeon hole argument,
we have that
the agents counter for its correct color is at least 
\[ t- (k-1)T \geq T. \]
Thus a correct color was correctly declared winning.
Now suppose the agent is close to the color border $(c_i,c_{i+1})$. Similarly as before the two counters have a value of at least 
\[ t- (k-2)T \geq 2T. \]
Thus at least one of the counters has a value larger than $T$ and was correctly declared winning.

By union bound, this holds for all columns w.h.p. since there are at most $n$ columns.



\end{proof}

%% file: missing.tex
\section{Missing Proofs}\label{sec:missing}

\begin{proof}[Proof of \autoref{thm:gradientlower}]
Fix an arbitrary grid $G$ with dimensions $a\times b$.
We will construct a grid $G'$
of dimensions $a' \times b'$ such that there exists a node (agent)
$u$ of $G$ with coordinates $(x,y)$ and a node 
$u'$ of $G'$ with coordinates $(x',y')$ receiving identical measured gradients all sources are at exactly the same distance.
$u$ needs to be white in $G$  (blue in $G'$, respectively) in order for the flag to be $\varepsilon$-approximation. 
Since the settings are indistinguishable, no algorithm achieving a $\varepsilon$-approximation is possible.

For the proof we assume w.l.o.g $a > b$, the proof for $a < b$ is analogous.

We fix the points we consider halfway up, such that $y = b/2$ and $y' = b'/2$. Our goal is to choose $(x', y')$
such that the following constraints are fulfilled 

\begin{equation} \label{eq:d1}
    x^2 + \frac{b^2}{4} = x'^2 + \frac{b'^2}{4}
\end{equation}

\begin{equation} \label{eq:d2}
    (x-a)^2 + \frac{b^2}{4} = (x'-a')^2 + \frac{b'^2}{4}
\end{equation}
\vspace{1mm}

ensuring that the distances to the corresponding sources are identical. Therefore, with this distance information alone, nodes in the two settings cannot distinguish between the two settings. We include the following two constraints to ensure that nodes in both settings must choose different colors, even in $\varepsilon$-approximate flags 

\begin{equation}\label{eq:cow}
    x = \left(\frac{1}{3}+\epsilon\right) a
\end{equation}

\begin{equation}\label{eq:cat}
    x' = \left(\frac{1}{3}-\epsilon\right) a'
\end{equation}
\vspace{1mm}

We solve this equation by showing that there exist indeed points $(x',y')$ of a $a'\times b'$ flag such that all of the equations are fulfilled with physically meaningful values, implying that the coloring decision based on these distances alone cannot be size invariant.

Subtracting  \eqref{eq:d2} from \eqref{eq:d1} gives
$
    x^2 - (x-a)^2 = x'^2 - (x'-a')^2
$, 
which is equivalent to
$
    2ax - a^2 = 2a'x' - a'^2
$.
We can now plug in \eqref{eq:cow} and \eqref{eq:cat} into this and derive

\begin{align*}
    2ax - a^2 &= 2a'x' - a'^2\\
    \Leftrightarrow
     2a^2\left(\frac{1}{3}+\epsilon\right) - a^2
     &= 
        2a'^2\left(\frac{1}{3}-\epsilon\right) - a'^2\\
       \Leftrightarrow
     a^2\left(-\frac{1}{3}+2\epsilon\right)
     &= 
         a'^2\left(-\frac{1}{3}'2\epsilon\right)\\  
\end{align*}
Thus we can choose
$
 a' = a \sqrt{ \frac{1/3-2\epsilon}{1/3+2\epsilon} }
$, which is positive for $\varepsilon < 1/6$. Plugging this into \eqref{eq:cat} gives
 $x' = \left(\frac{1}{3}-\epsilon\right)  \sqrt{ \frac{1/3-2\epsilon}{1/3+2\epsilon} }$, which is also positive for $\varepsilon < 1/6$.
 Finally we have
 \[ b' = 4x^2 + b^2 -4x'^2 = 4 \left( \left(\frac13+\epsilon\right) a \right)^2 -
 4\left( \left(\frac13-\epsilon\right)  \sqrt{ \frac{1/3-2\epsilon}{1/3+2\epsilon} } a  \right)^2    
 +     b^2 >0. \]
 \vspace{1mm}
 
Thus, putting everything together, we get valid solutions
for $x',y',a'$ and $b'$ yielding the claim.

\end{proof}